\documentclass[prl,aps,twocolumn,twoside,superscriptaddress,nofootinbib,showpacs]{revtex4}
\usepackage{times}
\usepackage{graphicx,epic,eepic,epsfig,amsmath,latexsym,amssymb,verbatim,revsymb,color}
\usepackage{array}
\usepackage{xspace}
\usepackage{bbm}
\usepackage{hyperref}
\usepackage{theorem}

\newtheorem{definition}{Definition}

\newtheorem{lemma}[definition]{Lemma}

\newtheorem{theorem}[definition]{Theorem}

\newenvironment{remark}{\noindent \textbf{{Remark~}}}{}

\def\squareforqed{\hbox{\rlap{$\sqcap$}$\sqcup$}}
\def\qed{\ifmmode\squareforqed\else{\unskip\nobreak\hfil
\penalty50\hskip1em\null\nobreak\hfil\squareforqed
\parfillskip=0pt\finalhyphendemerits=0\endgraf}\fi}
\def\endenv{\ifmmode\;\else{\unskip\nobreak\hfil
\penalty50\hskip1em\null\nobreak\hfil\;
\parfillskip=0pt\finalhyphendemerits=0\endgraf}\fi}
\newlength{\blank}
\newenvironment{proof}[1][{\hspace{-\blank}}]{{\noindent\textbf{Proof~{#1}.\ }}}{\hfill\qed\vskip 0.5\baselineskip}

\DeclareRobustCommand\idop{%
	\leavevmode\hbox{\small1\normalsize\kern-.33em1}}

\DeclareRobustCommand\scriptidop{\leavevmode\hbox{\fontsize{7}{8}\selectfont 1\scriptsize\kern-.33em1}}

\mathchardef\ordinarycolon\mathcode`\:
\mathcode`\:=\string"8000
\def\vcentcolon{\mathrel{\mathop\ordinarycolon}}
\begingroup \catcode`\:=\active
  \lowercase{\endgroup
  \let :\vcentcolon
  }

\newcommand{\nc}{\newcommand}
\nc{\rnc}{\renewcommand}
\nc{\beq}{\begin{equation}}
\nc{\eeq}{{\end{equation}}}
\nc{\beqa}{\begin{eqnarray}}
\nc{\eeqa}{\end{eqnarray}}
\nc{\lbar}[1]{\overline{#1}}
\nc{\bra}[1]{\langle#1|}
\nc{\ket}[1]{|#1\rangle}
\nc{\ketbra}[2]{|#1\rangle\!\langle#2|}
\nc{\braket}[2]{\langle#1|#2\rangle}
\nc{\inp}[2]{\langle#1|#2\rangle}
\nc{\proj}[1]{| #1\rangle\!\langle #1 |}
\nc{\avg}[1]{\langle#1\rangle}
\nc{\Rank}{\operatorname{Rank}}
\nc{\smfrac}[2]{\mbox{$\frac{#1}{#2}$}}
\nc{\tr}{\operatorname{Tr}}
\nc{\ox}{\otimes}
\nc{\Ox}{\bigotimes}
\nc{\dg}{\dagger}
\nc{\dn}{\downarrow}
\nc{\cA}{{\cal A}}
\nc{\cB}{{\cal B}}
\nc{\cC}{{\cal C}}
\nc{\cD}{{\cal D}}
\nc{\cE}{{\cal E}}
\nc{\cF}{{\cal F}}
\nc{\cG}{{\cal G}}
\nc{\cH}{{\cal H}}
\nc{\cI}{{\cal I}}
\nc{\cJ}{{\cal J}}
\nc{\cK}{{\cal K}}
\nc{\cL}{{\cal L}}
\nc{\cM}{{\cal M}}
\nc{\cN}{{\cal N}}
\nc{\cO}{{\cal O}}
\nc{\cP}{{\cal P}}
\nc{\cR}{{\cal R}}
\nc{\cS}{{\cal S}}
\nc{\cT}{{\cal T}}
\nc{\cX}{{\cal X}}
\nc{\cZ}{{\cal Z}}
\nc{\csupp}{{\operatorname{csupp}}}
\nc{\qsupp}{{\operatorname{qsupp}}}
\nc{\var}{\operatorname{var}}
\nc{\rar}{\rightarrow}
\nc{\lrar}{\longrightarrow}
\nc{\polylog}{\operatorname{polylog}}
\nc{\1}{{\openone}}
\nc{\sign}{{\operatorname{sign}}}

\def\e{\epsilon}

\nc{\RR}{{{\mathbb R}}}
\nc{\CC}{{{\mathbb C}}}
\nc{\FF}{{{\mathbb F}}}
\nc{\NN}{{{\mathbb N}}}
\nc{\ZZ}{{{\mathbb Z}}}
\nc{\PP}{{{\mathbb P}}}
\nc{\QQ}{{{\mathbb Q}}}
\nc{\UU}{{{\mathbb U}}}
\nc{\EE}{{{\mathbb E}}}
\nc{\id}{{\operatorname{id}}}

\nc{\olambda}{{\overline{\lambda}}}
\nc{\ulambda}{{\underline{\lambda}}}

\nc{\be}{\begin{equation}}
\nc{\ee}{{\end{equation}}}
\nc{\bea}{\begin{eqnarray}}
\nc{\eea}{\end{eqnarray}}
\nc{\<}{\langle}
\rnc{\>}{\rangle}
\nc{\Hom}[2]{\mbox{Hom}(\CC^{#1},\CC^{#2})}
\nc{\rU}{\mbox{U}}

\nc{\ob}[1]{#1}

\newcommand{\assign}{\ensuremath{\kern.5ex\raisebox{.1ex}{\mbox{\rm:}}\kern -.3em =}}

\begin{document}

\title{Are random pure states useful for quantum computation?}

\author{Michael J. Bremner} 
\affiliation{Department of Computer Science, University of Bristol, Bristol BS8 1UB, U.K.}
\email{michael.bremner@bris.ac.uk}
%\affiliation{Brizzle}

\author{Caterina Mora}
\affiliation{Institute for Quantum Computation, University of Waterloo, 200 University Ave. W. N2L 3G1, Canada}
%\affiliation{Some cold boring place}

\author{Andreas Winter}
\affiliation{Department of Mathematics, University of Bristol, Bristol BS8 1TW, U.K.}
\affiliation{Centre for Quantum Technologies, National University of Singapore,
 2 Science Drive 3, Singapore 117542}
%\email{a.j.winter@bris.ac.uk}
%\affiliation{World's airline industry}

\date{13 December 2008}

\begin{abstract}
We show the following: a randomly chosen pure state
as a resource for measurement-based quantum computation,
is -- with overwhelming probability -- of no greater help
to a polynomially bounded classical control computer, than
a string of random bits. Thus, unlike the familiar ``cluster
states'', the computing power of a classical control device is not increased from $\mathrm{P}$
to $\mathrm{BQP}$, but only to $\mathrm{BPP}$. 
The same holds if the task is to sample from a distribution rather 
than to perform a bounded-error computation. 
Furthermore, we show that our results can be extended to states 
with significantly less entanglement than random states.
\end{abstract}

\maketitle

In measurement based (or ``one-way'') quantum computation, two very different resources are used: one is a multi-qubit state $\ket{\Psi}$, the other is a classical algorithm used to determine how to measure the qubits, in which order and in which local basis~\cite{RaussendorfBriegel}. This clear separation of quantum and classical resources gives rise to the question: Which combinations of quantum states and classical control algorithms yield an advantage over classical computation?

In this paper we show that the efficiency requirements on classical processing of measurement data in measurement based models severely limits the class of quantum states which offer a computational speed-up over classical computers.  In particular, we demonstrate that the set of languages that can be decided by randomly chosen pure states together with polynomial-sized classical control circuits is the same (with high probability) as the set of languages that could be decided by polynomial-sized classical circuits and classical randomness alone (that is $\mathrm{BPP}$). Our intuition is that random pure states simply have too many uncorrelated parameters to allow for a computational speed-up over classical processors.  In support of this intuition we extend our main theorem to cover states which do not share the entanglement properties of typical states.

Much of the literature has focused on identifying particular states, or classes of states, for which universal quantum computation can be performed by utilizing a small set of single-qubit measurements and a simple classical control algorithm. This is generally done by recognizing certain ``nice" properties of a state which allow measurement outcomes to be interpreted as having applied a quantum gate to some predefined input state.

One can take a constructive approach, such as in~\cite{Gross06, Gross08} where the authors use techniques for the classical simulation of quantum systems to find simple rules that describe the effects of certain measurements. These rules apply to a wide-range of entangled states and can be used to show that a large variety of systems can support measurement based quantum computation. From a more physical perspective, other work has considered how ground/thermal states of natural systems can be used for measurement based quantum computation~\cite{Barrett08,Brennen08,Doherty08}.

Alternatively, one can identify general physical requirements that must be satisfied in order for it to be universal for quantum computation~\cite{maarten07,vandenNestDuerMiyakeBriegel}. In these papers the authors demonstrate that if the amount of entanglement in a family of states does not grow sufficiently quickly with the number of qubits, then there is no deterministic LOCC protocol that can prepare a family of cluster states.

The line of thinking in our paper is more in the vein of \cite{Anders08} where the authors examine how classical control computers of varying computational power are boosted by the addition quantum resources. For instance they demonstrate GHZ states enhance classical control devices which are only capable of calculating parities to $\mathrm{BPP}$.

Very recently Gross, Flammia and Eisert~\cite{Obergurgl-talk} have also shown, like in the current paper, that random states (in fact \emph{highly entangled} states) cannot be used for universal measurement-based quantum computation. They demonstrate this by proving that for certain problems in $\mathrm{BQP}$ which are thought to not be in $\mathrm{BPP}$, highly entangled states offer no advantage over classical randomness even given an oracle which supplies the ``best" set of single-qubit measurements to be performed.

\medskip
\noindent
\emph{Abstract measurement-based computation model.---} We begin our analysis by defining the following general model of computation - one which seems to capture all computationally efficient possibilities of using measurements on a quantum state to drive a computation.

\begin{definition}
  \label{def:model}
  A model of abstract measurement-based quantum computation (AMBQC) is
  a sequence of pairs $(\ket{\Psi_n},C_n)$ (with $n \rar \infty$),
  where $\ket{\Psi_n}$ is a $q(n)$-qubit quantum state and
  $C_n$ is a classical Boolean circuit on $w(n)$ bits and having
  at most $v(n)$ logical gates of up to $3$ bits, with distinguished 
  (multi-bit) registers as follows: an $n$-bit input register $x$, 
  a single bit output register $y$, a register $k$ to hold an integer in $[0;q(n)]$,
  a register of $O(q(n))$ bits to hold the measurement outcomes
  $m = m_1 \ldots m_{q(n)}$ of all the qubits measured,
  a sufficiently large output register $\alpha$ to describe the next
  qubit measurement $(L^{(\alpha)}_\mu)_\mu$, and finally workspace
  (ancilla bits) $a$.
  
  In practice we ask for $q(n)$, $v(n)$ and $w(n)$ to be polynomially
  bounded, but it will turn out that it is enough to require that
  they are not ``too big''.
\end{definition}

%\begin{definition}
%  \label{def:model}
%  A model of abstract measurement-based quantum computation (AMBQC) is
%  a sequence of pairs $(\ket{\Psi_n},C_n)$ (with $n \rar \infty$),
%  where $\ket{\Psi_n}$ is a $q(n)$-qubit quantum state and
%  $C_n$ is a classical Boolean circuit on $w(n)$ bits and having
%  at most $v(n)$ logical gates of up to $3$ bits, with distinguished 
%  (multi-bit) registers as follows: 
%  \begin{itemize}
%  \item an $n$-bit input register $x$, 
%  \item  a single bit output register $y$, 
%  \item a register $k$ to hold an integer in $[0;q(n)]$,
%  \item  a register of $O(q(n))$ bits to hold the measurement outcomes
%  $m = m_1 \ldots m_{q(n)}$ of all the qubits measured,
%  \item  a sufficiently large output register $\alpha$ to describe the next
%  qubit measurement $(L^{(\alpha)}_\mu)_\mu$, 
%  \item and finally a workspace (ancilla bits) $a$.
%  \end{itemize}
%  
%  In practice we ask for $q(n)$, $v(n)$ and $w(n)$ to be polynomially
%  bounded, but it will turn out that it is enough to require that
%  they are not ``too big''.
%\end{definition}

\noindent
A pair $(\ket{\Psi},C) = (\ket{\Psi_n},C_n)$ with certain $q=q(n)$,
$v=v(n)$ and $w=w(n)$ gives rise to a probabilistically branching
history of measurements as follows. Starting with the registers prepared
as $[x,y=0,k=0,m=0,\alpha=0,a=0]$, we run the circuit $C$ to obtain
$1\leq \ell_1 \leq q$ in the $k$-register and $\alpha_1$: this is interpreted 
as an instruction to measure qubit $\ell_1$ with the POVM $L^{(\alpha_1)}$
-- let the (probabilistic) outcome be $m_1$. After measuring $k$ qubits,
having obtained outcomes $m_1,\ldots,m_k$ the circuit $C$ is run on
$[x,y=0,k,m=m_1\ldots m_k 0\ldots 0,\alpha=0,a=0]$ to obtain
in the $k$-register an integer $1\leq \ell_{k+1} \leq q$ and
an $\alpha_{k+1}$. This tells us now to measure qubit $\ell_{k+1}$
with the POVM $L^{(\alpha_{k+1})}$, yielding another measurement
outcome $m_{k+1}$.
This iterates until $k=q$ is encountered; in that case, the reading
of the $y$-register in $C[x,y=0,k=q,m=m_1\ldots m_q,\alpha=0,a=0]$ is
the output of the computation. The probability that $y=1$ over all histories
(i.e.~the probability that the computation accepts) is denoted $C_x(\Psi)$.

\medskip
\begin{remark}
Note that our state $\ket{\Psi}$ has $q$ qubits and exactly $q$ measurements
are made. We shall from now on implicitly restrict to AMBQCs $(\ket{\Psi},C)$
in which all histories end up measuring all $q$ qubits (or equivalently,
there is no actually occurring history where some qubit is measured
twice). An AMBQC obeying this condition we call \emph{complete}; it is 
naturally fulfilled in all known specific models.
\end{remark}
\medskip

We say that $(\ket{\Psi},C) = (\ket{\Psi_n},C_n)$ computes a (partial)
function $f:\{0,1\}^n \rar \{0,1\}$ with bounded error, if
\[
  f(x)=1 \Rightarrow C_x(\Psi) \geq 2/3,\quad
  f(x)=0 \Rightarrow C_x(\Psi) \leq 1/3.
\]

Note that, even though in practice this will be an important restriction,
we impose no uniformity on the $C_n$, nor on the states $\ket{\Psi_n}$.

If we are interested in collective properties of all AMBQCs with all possible
inputs (as we shall be shortly), we may even disregard the $n$-bit input
$x$, as a slightly longer control circuit starting off in the all-zero
input can first prepare the input $x$ and then do the actual computation
described above.

These two points mean that we shall actually only look at particular
finite sized $n$, $q$, $v$ and $w$.

\medskip
\noindent
\emph{Random states are not universal.---}
As promised, the question we want to address is whether a generic (i.e., randomly chosen)
state $\ket{\Psi}$ is of any good use to an AMBQC? The way we think of
this is a little different from the usual MBQC, where the resource state
can typically be prepared easily in a quantum computer -- since random
states have enormous time complexity to prepare in a quantum computer~\cite{Mora:QKC}
we think of $\ket{\Psi}$ as being handed to us by an all-powerful, 
Merlinesque character. Since we are similarly not even able to
study a description of the state (as it is too long to read in time
polynomial in $n$~\cite{Mora:QKC}), we cannot be expected to come
up with the control circuit $C$ on our own. Instead it is described
to us by a helpful Merlin, too, giving us the
control circuit $C$ that best exploits the properties of $\ket{\Psi}$.

In simple terms, our main result states that for a typical random
states there is \emph{no} short control circuit that can do anything with
$\ket{\Psi}$ which cannot be simulated to sufficient precision using
classical random bits. 

\begin{theorem}
\label{thm:main}
  For a random state $\ket{\Psi}$ on $n$ qubits, consider classical Boolean
  control circuits $C$ of width $w$ and having at most $v$ gates,
  let $C(\Psi)$ be the probability of acceptance of the AMBQC $(\ket{\Psi},C)$,
  and similarly $C(2^{-q}\1)$ the probability of acceptance when instead
  of $\Psi$ the maximally mixed state $2^{-q}\1$ is used. Then, for
  any $\epsilon>0$,
  \begin{equation}
    \label{eq:main}
    {\Pr}_{\Psi} \!\left\{ \exists C\ \bigl| C(\Psi)-C(2^{-q}\1) \bigr| > \epsilon \right\}\!
                    \leq \!\bigl( 8^8 w\bigr)^{3v} \, e^{-c\epsilon^2 2^q}\!\!,
  \end{equation}
  where $c = \frac{1}{9\pi^3}$ is a universal constant. (Observe that
  the existential quantifier implicitly restricts to complete AMBQCs.)
  \end{theorem}
  
So, whenever $v\ln w = o(2^q)$ -- e.g.~for all polynomially bounded
circuits -- the right hand side of eq. (\ref{thm:main}) goes to zero exponentially, and hence
for most states the measurement results coming from $\Psi$ can
be replaced by classical independent randomness: this changing the
acceptance probability by at most $\epsilon$, regardless of the
circuit used.

\medskip
\begin{proof}
For a given two-outcome POVM with operators $P\geq 0$ and
$Q=\1-P \geq 0$ acting on $(\mathbbm{C}^2)^{\otimes q}$, a straightforward
application of Levy's Lemma~\cite{MilmanSchechtman} yields
\begin{equation}
  \label{eq:levy-lalala}
  {\Pr}_{\Psi} \left\{ \big| \bra{\Psi} P \ket{\Psi} - 2^{-q} \tr P \bigr| > \epsilon \right\}
                \leq 4e^{-c\epsilon^2 2^q},
\end{equation}
where $c=\frac{1}{9\pi^3}$. The reason is that for any $0\leq P\leq \1$, the
function $\ket{\Psi} \mapsto \bra{\Psi} P \ket{\Psi}$ has Lipschitz constant $1$.
[For a real-valued function $f$ with Lipschitz 
constant $\Lambda$, Levy's Lemma bounds the 
probability that $|f(x)-\EE f|>\epsilon$, for a uniformly
random point $x$ on the $(d-1)$-dimensional Euclidean 
sphere, by $4\e^{-c\epsilon^2d/\Lambda^2}$~\cite{MilmanSchechtman}. We apply this to $f(\ket{\Psi})=\bra{\Psi} P \ket{\Psi}$, and $d=2.2^q$.]

Now, observe that every control circuit effectively describes such a two-outcome
POVM: the circuit starts making measurements on the system, and for each
sequence of previous outcomes decides on the next measurement; at the end,
the complete data obtained -- the sequence $\ell = \ell_1\ldots\ell_q$ of qubits
measured, the local measurements $\alpha=\alpha_1\ldots\alpha_q$ and the
outcomes $m=m_1\ldots m_q$ -- is used to decide acceptance or rejection.
Thus, we find the accepting and rejecting operators,
\[
  P = \sum_{(\ell,\alpha,m) \atop \text{~acc.~history~}} 
            \bigotimes_{k=1}^q \left( L^{(\alpha_k)}_{m_k} \right)^{\ell_k},
\]
and $Q=\1-P$. In this way, clearly, $C(\Psi) = \bra{\Psi} P \ket{\Psi}$.

The number of possible circuits to consider is at most
$\left( 8^8 {w \choose 3} \right)^{v} \leq \frac{1}{6} \bigl( 8^8 w\bigr)^{3v}$,
so we put together eq.~(\ref{eq:levy-lalala}) with the union bound to obtain
eq.~(\ref{eq:main}), observing the simple equation $C(2^{-q}\1) = 2^{-q} \tr P$.

Finally, we have to explain why the latter probability can be
sampled efficiently classically. But that is straightforward, since
the maximally mixed state $2^{-q}\1$ is a tensor product of single-qubit
maximally mixed states $\frac{1}{2}\1$, so indeed each measurement
result of a local POVM $(L_\mu)_\mu$ may be sampled independently
with probability $\frac{1}{2}\tr L_\mu$ for outcome $\mu$,
which can be done efficiently thanks to the classical description of the POVM.
\end{proof}

\begin{remark}
  For traditional MBQC, the local measurements are simply von Neumann
  measurements, i.e.~consisting of two orthogonal basis projectors.
  In that case the measurement outcomes are simply replaced by
  independent random bits.

  It is clear that the above can be generalized without any difficulty
  to qu$d$its as elementary systems. Equivalently, models that consider 
  measurements on bounded-sized sets of qubits could also be considered.
  
  Note furthermore that the step-by-step simulation above produces a probability
  distribution over the same computational histories as the original AMBQC.
  We do not claim, however, that these two distributions are close (which
  isn't true in general), but only that the efficient coarse-grainings
  represented by the output bit $y$ are.
\end{remark}

\medskip
\noindent
\emph{Sampling of a t-bit string.---}
If the object of the computation is to produce a sample from
distribution on, say, $t$-bit strings, we denote by $C(\Psi)$ the
resulting distribution.
If $t\ll q$ we can apply Levy's Lemma to all $2^t$ probability values of
$C(\Psi)$ -- and in a generalization of Theorem~\ref{thm:main} we
would like to show that for most states $\Psi$ it is indistinguishable
from $C(2^{-q}\1)$, i.e.~the distribution obtained from running
the AMBQC on the maximally mixed state. Since
\[
  \bigl\| C(\Psi)-C(2^{-q}\1) \bigr\|_1 = \sum_{y=0}^{2^t-1} \bigl| C(\Psi)[y]-C(2^{-q}\1)[y] \bigr|,
\]
for the left hand side to exceed $\epsilon$ requires the right hand side to
have one term exceeding $\epsilon/2^t$. Via the same counting argument
as in Theorem~\ref{thm:main}, enhanced by an additional union over
all sample strings $y$, we obtain
\begin{equation}\begin{split}
  \label{eq:main-sampling}
  {\Pr}_{\Psi} &\left\{ \exists C\ \bigl\| C(\Psi)-C(2^{-q}\1) \bigr\|_1 > \epsilon \right\} \\
               &\phantom{===========}     
                                 \leq 2^t \bigl( 8^8 w\bigr)^{3v} \, e^{-c\epsilon^2 2^{q-2t}},
\end{split}\end{equation}
where $c = \frac{1}{9\pi^3}$ is as before. 

In other words, as long as $2^{2t} t v\ln w = o(2^q)$, it is exponentially
unlikely for a random state to provide any advantage for AMBQC over
a maximally mixed state.
We note that this condition is typically fulfilled in ``traditional''
cluster state models, where both $t$ and the depth of the
quantum circuit are polynomial in the input size $n$, so $q$ is a higher order 
polynomial in $n$ than $t$.

\medskip
\noindent
\emph{Schmidt-rank-$K$ states.---}
It is natural to wonder which exact property makes a random state so
particularly useless for AMBQC. Two answers might come to mind: first,
random states have, with high probability, almost maximal description
complexity~\cite{Mora:QKC}. Another is that typical random
states are highly entangled: indeed, 
Gross \emph{et al.}~\cite{Obergurgl-talk} show that the
geometric measure of entanglement on $q=q(n)$ qubits,
\[
  E_g(\ket{\Psi}) = -\log \max_{\ket{\varphi} = \bigotimes_j \ket{\varphi^{(j)}}} | \bra{\varphi} \Psi \rangle |^2,
\]
is with high probability $\geq q - 2\log q - O(1)$. Then they show (similar
to our approach above) that in performing a computation with only one-sided
and bounded error, the measurement outcomes of such states may be 
replaced by independent random bits. The resulting probabilistic
computation still has bounded, one-sided error.
\\
This motivates the following definition and theorem.

\begin{definition}[Random Schmidt-rank $K$ states]
  \label{defi:schmidt}
  Construct the following random state $\Psi$ on $q=q(n)$ qubits, 
  called \emph{random Schmidt-rank $K$ state}.
  We define its distribution by a sequence of random experiments: let
  \begin{equation}
    \label{eq:R}
    R \assign \sum_{j=1}^{K} \proj{\psi_j^{(1)}} \ox \cdots \ox \proj{\psi_j^{(q)}},
  \end{equation}
  where all the $qK$ unit vectors $\ket{\psi_j^{(\ell)}}$ are chosen independently 
  at random from any measure on the pure states of 
  $\CC^2$ such that $\EE \psi_j^{(\ell)} = \frac{1}{2}\1$.
  Now pick a unit vector $\ket{\Psi_0}$ from the support of $R$ according
  to the unitary invariant measure, and finally let
  \begin{equation}
    \label{eq:schmidt}
    \ket{\Psi} = \frac{1}{\sqrt{\bra{\Psi_0} R \ket{\Psi_0}}} \sqrt{R}\ket{\Psi_0}.
  \end{equation}
\end{definition}

\begin{theorem}
  \label{thm:rank-K}
  For a random Schmidt-rank $K$ state $\ket{\Psi}$ on $q$ qubits, 
  with $64 \leq K \leq 2^q$ (which implies $q \geq 6$),
  consider classical Boolean
  control circuits $C$ of width $w$ and having at most $v$ gates.
  Then, for any $\epsilon>0$,
  \begin{equation}\begin{split}
    \label{eq:rank-K}
    {\Pr}_{\Psi} &\left\{ \exists C\ \bigl| C(\Psi)-C(2^{-q}\1) \bigr| > \epsilon \right\}  \\
                 &\phantom{========}
                  \leq \left( 2^q + \bigl( 8^8 w\bigr)^{3v} \right) e^{-c'\epsilon^2 K^{1/3}},
  \end{split}\end{equation}
  and where $c' = \frac{1}{1296\pi^3}$ is a universal constant.
\end{theorem}

In other words, whenever $q + v\ln w = o\bigl(K^{1/3}\bigr)$ -- e.g.~for 
all polynomially bounded circuits and superpolynomial $K$ -- the right 
hand side of eq. (\ref{thm:rank-K}) goes to zero exponentially, and hence for most Schmidt-rank $K$ 
states the measurement results coming from $\Psi$ can
be replaced by classical independent randomness: this changes the
acceptance probability by at most $\epsilon$, regardless of the
circuit used.

To prove this we shall use Levy's Lemma~\cite{MilmanSchechtman} once again,
but we also need two further concentration results:
\begin{lemma}
  \label{lemma:R}
  For the random operator $R$ in eq.~(\ref{eq:R}) such that $K \geq 4.2^k$
  and $2 \leq k \leq q$,
  \[
    {\Pr}_R\left\{ \bigl\| R \bigr\|_\infty > 2\frac{K}{2^k} \right\} \leq 2^q e^{-K 2^{-k}/3}.
  \]
\end{lemma}
\begin{proof}
  We start by observing, for the reduction of $R$ onto the first $k$ qubits,
  \[
    R^{(k)} \assign \tr_{k+1\ldots q} R = \sum_{j=1}^K \proj{\psi_j^{(1)}} \ox \cdots \ox \proj{\psi_j^{(k)}},
  \]
  that $\| R \|_\infty \leq \bigl\| R^{(k)} \bigr\|_\infty$. This follows from
  the result in~\cite{NielsenKempe} that a separable operator $R$
  majorizes all its reductions, applied to the largest eigenvalue.
  Hence, we only need to bound the probability that
  $\bigl\| R^{(k)} \bigr\|_\infty$ is ``large''.
  
  Now, $R^{(k)} = \sum_{j=1}^K X_j$ is the sum of i.i.d.~random operators
  $X_j = \proj{\psi_j^{(1)}} \ox \cdots \ox \proj{\psi_j^{(k)}} \in [0;\1]$,
  so the theory of large deviations of operator valued random variables
  from~\cite{AhlswedeWinter:cq-ID} applies. Since $\EE X_j = 2^{-k}\1$,
  we can use~\cite[Thm.~19]{AhlswedeWinter:cq-ID} directly, and get
  \[\begin{split}
    {\Pr}_R\left\{ \bigl\| R^{(k)} \bigr\|_\infty > 2\frac{K}{2^k} \right\} 
                                             &\leq 2^k e^{-K D(2.2{-k}\|2^{-k})}    \\
                                             &\leq 2^q e^{-K 2^{-k}/3},
  \end{split}\]
  observing the elementary inequality
  \(
    D\bigl(2.2^{-k}\|2^{-k}\bigr) \geq (2\ln 2 - 1)2^{-k} \geq 2^{-k}/3
  \)
  for the relative entropy, as well as $k\leq q$.
\end{proof}

\medskip
\begin{remark}
  The bound on $\| R \|_\infty$ in Lemma~\ref{lemma:R} is in general an
  overestimate. Indeed, if $K \leq \epsilon 2^{q/2}$, it is straightforward
  to see that $\EE \tr R^2 \leq K + K^2 2^{-q} \leq K + \epsilon^2$.
  Elementary arguments, using $\tr R^2 \geq \tr R = K$ show that then,
  with high probability, $\| R \|_\infty \leq 1+O(\epsilon)$.
\end{remark}

\begin{lemma}%[Hoeffding bound]
  \label{lemma:hoeffding}
  For the random operator $R$ in eq.~(\ref{eq:R}), and $0\leq P \leq \1$,
  \[
    {\Pr}_R \left\{ \left| \frac{1}{K}\tr RP - \frac{1}{2^q} \tr P \right| > \epsilon \right\}
                                                                     \leq 2 e^{-2 \epsilon^2 K}.
  \]
\end{lemma}
\begin{proof}
  Observe that
  \[
    \tr RP = \sum_{j=1}^K \tr \bigl(\proj{\psi_j^{(1)}} \ox \cdots \ox \proj{\psi_j^{(q)}} P\bigr)
  \]
  is a sum of $K$ i.i.d.~real random variables $X_j \in [0;1]$, and in the
  lemma we are looking at a large deviation of their empirical mean from
  the expectation, $\EE X_j = 2^{-q}\tr P$.
  Hence the classical Hoeffding bound~\cite{Hoeffding} applies:
  \[
    \Pr\left\{ \left| \frac{1}{K}\sum_j X_j - \EE X_1 \right| > \epsilon \right\}
                                                                \leq 2 e^{-2\epsilon^2 K},
  \]
  and we are done.
\end{proof}

\medskip\noindent
We are now ready for the 
\begin{proof}[of Theorem~\ref{thm:rank-K}]
  Picking the random state $\ket{\Psi}$, we have implicitly already constructed
  the operator $R$ in Definition~\ref{defi:schmidt}.

  First, according to Lemma~\ref{lemma:R}, and choosing
  $k = \left\lfloor \frac{2}{3}\log K \right\rfloor$, we get
  \begin{equation}
    \label{eq:R-bound}
    \| R \|_\infty \leq 4 K^{1/3},
  \end{equation}
  except with probability $\leq 2^q e^{-K^{1/3}/3}$.

  Second, according to Lemma~\ref{lemma:hoeffding}, we have
  for all measurement POVMs $(P,\1-P)$ constructed by the allowed
  classical control circuits [of which there are
  $M \leq \frac{1}{6}\bigl( 8^8 w\bigr)^{3v}$ -- see the proof of
  Theorem~\ref{thm:main}],
  \begin{equation}
    \label{eq:projection}
    \left| \frac{1}{K}\tr RP - 2^{-q}\tr P \right| \leq \epsilon,
  \end{equation}
  except with probability $\leq 2 M e^{-2\epsilon^2 K}$.
  
  Third, assuming eq.~(\ref{eq:R-bound}) holds for a particular $R$,
  application of Levy's Lemma~\cite{MilmanSchechtman} to the same POVM
  elements $P$ is possible, noting that the Lipschitz constant of the
  function $\ket{\Psi_0} \mapsto \tr \sqrt{R}\Psi_0\sqrt{R} P$ 
  is $\Lambda \leq 8 K^{1/3}$.
  We find that for all these $P$,
  \begin{equation}
    \label{eq:Psi0-P}
    \left| \tr \sqrt{R} \Psi_0 \sqrt{R} P - \frac{1}{K}\tr RP  \right| \leq \epsilon,
  \end{equation}
  except with probability $\leq 4 M e^{-c \epsilon^2 K/\Lambda^2}$,
  with $c=\frac{1}{9\pi^3}$, as in Theorem~\ref{thm:main}.
  The special case $P=\1$ is trivially included:
  \begin{equation}
    \label{eq:Psi0-I}
    \left| \tr \sqrt{R} \Psi_0 \sqrt{R} - 1 \right| \leq \epsilon,
  \end{equation}
  i.e.~$\sqrt{R}\ket{\Psi_0}$ is already almost normalized.
  
  Putting the three steps together, we find that if 
  eqs.~(\ref{eq:R-bound}), (\ref{eq:projection}), (\ref{eq:Psi0-P})
  and (\ref{eq:Psi0-I}) hold,
  then for all eligible control circuits,
  \[
    \bigl| C(\Psi)-C(2^{-q}\1) \bigr| \leq  3\epsilon.
  \]
  As noted above, however, this will be the case
  except with probability bounded above by
  \[
    2^q e^{-K^{1/3}/3} + 2 M e^{-2\epsilon^2 K} + 4 M e^{-c \epsilon^2 K/\Lambda^2}.
  \]
  Redefining $\epsilon \mapsto \epsilon/3$ concludes the proof.
\end{proof}

It is intuitive (and not difficult to show) 
that with high probability, a random Schmidt-rank
$K$ state satisfies $E_g(\ket{\Psi}) \leq \log K + O(1)$, 
and since $\ket{\Psi}$ is always a superposition of $K$
product states, also the descriptive
(quantum Kolmogorov) complexity of the state is bounded by
an exponential in $K$ (this follows from a straightforward counting
argument, cf.~\cite{Mora:QKC}). Hence even these states, though
failing the criterion of~\cite{Obergurgl-talk}, are useless for AMBQC.
We would like to say that this is due to the complexity of a random
choice of pure state, but have to stress that it is not the
descriptive complexity of~\cite{Mora:QKC}. Rather, it is the fact
that all degrees of freedom given to the state are exhausted uniformly.

Note that the number of degrees of freedom sufficient for this is
anything growing superpolynomially in $n$, if the control circuit and
$q$ are polynomially bounded. But it is also necessary, because
if $K$ is polynomial, then $\ket{\Psi}$ always has an efficient
classical description, and so have all the states occurring through
the course of the computation; in other words, the state is useless
for AMBQC for another reason, as it is simulable in $\mathrm{P}$.

\medskip
\noindent
\emph{Conclusion.---} We have shown that for decision problems with bounded error probability (and more
generally for the task of approximately sampling a distribution on ``few'' bits),
a generic quantum state is (with overwhelming probability) not more useful as a resource
to a classical control mechanism for a generalized measurement-based model,
than a random bit string. The only condition on the classical control is that
it can be built as a Boolean circuit of subexponential depth.
In other words, unless $\mathrm{BQP}=\mathrm{BPP}$, such states won't yield
universal quantum computation when used in any reasonable environment controlling
the sequence of measurements.
However, the result is not limited to $\mathrm{BQP}$, it also encompasses promise problems,
as long as the AMBQC is supposed to be polynomially efficient and has bounded
error; furthermore the complexity may essentially be anything strictly
smaller than exponential. (Observe that an exponential classical control could
simulate the whole state, so its power is also not increased by having
access to $\ket{\Psi}$.)

Finally, even decidedly ``non-random'' states (in the sense that their
distribution is not unitary invariant) still have the same property
if only they are drawn from a large enough manifold, as we have demonstrated
with random states of bounded Schmidt rank.

\medskip
\noindent
\emph{Acknowledgments.---}
We acknowledge insightful and interesting discussions with 
David Gross, Steve Flammia and Jens Eisert;
%who have independently discovered similar, yet importantly different, results. 
we also thank Ashley Montanaro for important suggestions.
MB and AW are supported by the EC-FP6-STREP network QICS.
CM acknowledges support by the IQC, QuantumWorks NSERC Innovation Platform, 
and Ontario Centres of Excellence.
AW is supported by the U.K. EPSRC through the ``QIP IRC'' and an
Advanced Fellowship, by a Royal Society Wolfson Merit Award
and by the European Commission through IP ``QAP''. The Centre for
Quantum Technologies is funded by the Singapore Ministry of Education
and the National Research Foundation as part of the Research Centres
of Excellence programme.

\end{document}